\newcommand{\be}{\begin{equation}}
\newcommand{\ee}{\end{equation}}

\documentclass[journal,10pt, letterpaper]{IEEEtran}

\usepackage{booktabs}
\usepackage{psfrag}
\ifCLASSINFOpdf
   \usepackage[pdftex]{graphicx}

\else
  
   \usepackage[dvips]{graphicx}

\fi

\graphicspath{{./figures/}}

\usepackage[cmex10]{amsmath}
\usepackage{amssymb}
\usepackage{amsthm,epstopdf,nicefrac,pgfplots}

  \pgfplotsset{compat=newest} 
  \pgfplotsset{plot coordinates/math parser=false}

\usepackage{float}

\usepackage{rotating}

\usepackage{mathtools}

\usepackage{gensymb}

\renewcommand\appendix{\par
  \setcounter{section}{0}
  \setcounter{subsection}{0}
  \setcounter{figure}{0}
  \setcounter{table}{0}
  \renewcommand\thesection{Appendix \Alph{section}}
  \renewcommand\thefigure{\Alph{section}\arabic{figure}}
  \renewcommand\thetable{\Alph{section}\arabic{table}}
}

\usepackage{color}

\usepackage{algpseudocode}
\usepackage{algorithm}

\twocolumn

  \newlength\fheight 
    \newlength\fwidth 
\usepackage{balance}

\usepackage{algorithm}
\usepackage{algpseudocode}
\usepackage{url}

\usepackage[tight]{subfigure}
\usepackage[font=footnotesize,skip=2pt]{caption}

\usepackage{balance}

\newtheorem{theorem}{Theorem}

\begin{document}

\title{
Practical Issues in the Synthesis of\\ 
Ternary Sequences
}


\author{A.~De Angelis, 
    J.~Schoukens, 
    K.~R.~Godfrey,
    P.~Carbone% <-this % stops a space
\thanks{
This work was supported in part by the Fund for Scientific Research (FWO-Vlaanderen), by the Flemish Government (Methusalem), the Belgian Government through the Inter university Poles of Attraction (IAP VII) Program, and by the ERC advanced grant SNLSID, under contract 320378.
}
\thanks{A.~De Angelis and P.~Carbone are with the Engineering Department,  University of Perugia, via G. Duranti 93, 06125 Perugia, Italy.
\{alessio.deangelis, paolo.carbone\}@unipg.it}% <-this % stops a space
\thanks{J. Schoukens is with Department ELEC, Vrije Universiteit Brussel, Pleinlaan 2, B1050 Brussels, Belgium.}% <-this % stops a space
\thanks{K.~R.~Godfrey is with the School of Engineering, University of Warwick, CV4 7AL, Coventry, United Kingdom.}%
\thanks{
Reference to the final published version: A. De Angelis, J. Schoukens, K. R. Godfrey and P. Carbone, ``'Practical Issues in the Synthesis of Ternary Sequences,'' in IEEE Transactions on Instrumentation and Measurement, vol. 66, no. 2, pp. 212-222, Feb. 2017.
Link to the final published version: 10.1109/TIM.2016.2622778. $\copyright$ 2017 IEEE. Personal use of this material is permitted. Permission from IEEE must be obtained for all other uses, in any current or future media, including reprinting/republishing this material for advertising or promotional purposes, creating new collective works, for resale or redistribution to servers or lists, or reuse of any copyrighted component of this work in other works}
}


\markboth{Preprint version}{}

\maketitle

\begin{abstract} 
Several issues related to the practical synthesis of ternary sequences with specified spectra are addressed in this paper. Specifically, sequences with harmonic multiples of two and three suppressed are studied, given their relevance when testing and characterizing nonlinear systems.
In particular, the effect of non-uniform Digital to Analog Converter (DAC) levels on the spectral properties of the generated signal is analyzed.
It is analytically shown that the DAC non-uniform levels result in degraded harmonic suppression performance. 
Moreover, a new approach is proposed for designing ternary sequences, which is flexible and can be adapted to suit different requirements.
The resulting sequences, denoted as randomized constrained sequences, 
are characterized theoretically by deriving an analytical expression of the power spectral density. Furthermore, they are extensively compared with three synthesis approaches proposed in the literature.
The approach is validated by numerical simulations and experimental results, showing the potential to achieve harmonic suppression performance of approximately 100 dB.
\end{abstract}

\begin{IEEEkeywords}
Frequency response function measurement, ternary sequences, digital-to-analog converters, spectral analysis.
\end{IEEEkeywords}

%\psfrag{td}[lb][bl]{\small \hskip3cm $\theta/\Delta$}

% \fi
\IEEEpeerreviewmaketitle

%\balance

\section{Introduction}
When testing and characterizing nonlinear systems, such as large-signal communication amplifiers \cite{MirriEtAl2004}, an excitation signal is provided at the input of the system under test and the response at the output is measured. 
Typically, the testing process requires a careful design of the excitation signal. Moreover, the analysis of systems in the frequency domain may be more advantageous than a time-domain approach. 
When using the frequency domain based approach, periodic signals with a specified power spectrum provide considerable advantages over other types of excitation signals, such as transient or noise-like signals.
By properly choosing the spectral components of the periodic signal, in fact, it is possible to isolate, detect, and analyze quantitatively the nonlinear distortions, while from the same measurements also a nonparametric noise model is retrieved. 
This mitigates the impact of nonlinear distortions on frequency response function (FRF) measurement \cite{Pintelon&Schoukens2012}. 

In fact, nonlinear distortions create additional harmonics at the output that were not present at the input. 
Specifically, for periodic signals, a nonlinearity of degree $n$ generates additional harmonics at frequencies given by all the possible combinations of the harmonics of the input signal, taken $n$ at a time.
In the particular case of $n$ even, any combination of odd harmonics always gives even harmonics
Therefore, if a signal consisting only of odd harmonics is applied at the input of a nonlinear system, even-order nonlinearities do not influence the FRF measurement, because they do not overlap with the harmonics originally present in the input signal  \cite{Pintelon&Schoukens2012}.
However, with such a signal, it is impossible to completely eliminate the impact of odd-order nonlinearities, but the impact of third-order nonlinearities on FRF measurements is greatly reduced if harmonic multiples of three are excluded from the signal \cite{Barker&Godfrey1999}.
For these reasons, it is advisable to design excitation signals with harmonic multiples of both two and three suppressed. 
The simplest multi-level signals that allow suppression of such harmonics are ternary sequences and therefore they are the object of considerable interest by the research community \cite{BarkerEtAl2005}.

In this paper, the problem of designing excitation sequences is approached from a measurement point of view. Specifically, it is addressed by analyzing the relevant practical aspects and by proposing and characterizing a design strategy. 
Therefore, the aim of this paper is twofold. On the one hand, we provide an analysis of the implementation issues related to the generation and acquisition of ternary sequences. 
On the other hand, we propose a new method to design ternary sequences with harmonic multiples of two and three suppressed that will be more robust with respect to the effect of non-ideal levels in the Digital to Analog Converter (DAC) generator. 
The method is based on numerical optimization, and is tunable to meet a wide range of requirements for the excitation signal spectrum. 
We provide a theoretical characterization of the proposed method and describe an extensive experimental evaluation of its performance compared with other ternary sequence synthesis methods from the literature, extending the preliminary results in \cite{DeAngelisEtAl_I2MTC2016}.

\section{Related Work}
\label{sec:related_work}
The design of excitation signals with predefined properties has received considerable attention in the instrumentation and measurement literature. 
Several applications have been considered, including channel measurement for communication systems, electrical impedance tomography, nonlinear systems characterization, FRF measurement, analog-to-digital converters (ADC) testing, and microwave measurement techniques. 

In particular, in the recent paper \cite{ZhouEtAl2016}, the channel measurement method for long-term evolution communication networks is proposed and validated in a high-speed railway scenario. The method uses a periodic multifrequency signal as the cell-specific reference signal. This facilitates the measurement of the channel impulse response.

The design of excitation signals using orthogonal codes is proposed for multiple-input-multiple-output systems in \cite{GeversEtAl2015}. The proposed approach is implemented in an electrical impedance tomography application, allowing for faster operation with respect to the conventional time-division approach. 
Furthermore, the performance of several excitation signals for portable biomedical applications is compared in \cite{BouchaalaEtAl2015}.

A method for designing multisine signals with logarithmic distribution of spectral lines is proposed in \cite{GeerardynEtAl2013}. These signals provide robust estimation of the FRF of resonating systems over a wide frequency range. 
Moreover, algorithms for designing multisines with low crest factor are proposed and analyzed in \cite{Friese1997} and \cite{YangEtAl2015}.

The problem of measuring the nonlinear behavior of ADCs by proper synthesis of excitation signals using a multisine design approach has been studied in \cite{OngEtAl2012}. This approach, based on optimizing a phase-plane objective function, is shown to represent the conversion error distribution more accurately with respect to other multisine design techniques. 
In the same application area, a signal generation method, based on the linearity property of the sine function, is proposed in \cite{Vora&Satish2011} for measuring ADC static nonlinearity.

In the microwave measurement field, a multisine excitation design approach is proposed for large-signal S-parameter calibration over a wide frequency range in \cite{VanMoer&Rolain2008}. This approach reduces the time necessary for characterizing the nonlinear distortions of the system under test with respect to the conventional swept-sine approach. 

When designing excitation signals, discrete multilevel sequences are widely studied, because they are easy to generate in a practical setting, and their properties can be tuned for the considered application in a simple fashion \cite{Tan&Godfrey2002}. As an example, the design of discrete multilevel sequences for characterizing nonlinear systems is studied in \cite{WongEtAl2013}. It is shown that, by adjusting properly signal levels and associated probabilities, it is possible to mimic the properties of a Gaussian input sequence. This allows for reducing the bias in the measurement of the best linear approximation of nonlinear systems.

Among discrete multilevel sequences for testing nonlinear systems, ternary sequences are particularly important. 
As discussed in Section I, in fact, when measuring the FRF of nonlinear systems, it is advisable to suppress harmonic multiples of two and three.
Such suppression can only be achieved using signals having more than two levels \cite{TanEtAl2005}. Using binary signals, in fact, it is only possible to suppress even-order harmonics \cite{Barker&Godfrey1999}. For this reason, pseudorandom ternary sequences are interesting, since they represent the simplest form of multi-level signal and are easily applied to transducers and actuators \cite{BarkerEtAl2005}. Moreover, such signals, when compared to signals with a larger number of levels, normally provide better dispersion performance, as quantified by the performance index for perturbation signals introduced in \cite{GodfreyEtAl1999}, while being easier to design.

The state-of-the-art in the synthesis of ternary sequences is largely based on theoretical results.
The present paper extends these results 
%state-of-the-art outlined above 
because it considers 
the usage of ternary sequences in conjunction with a treatment of practical issues related to the synthesis of excitation signals and it proposes a new generation mechanism that is more robust with respect to DAC nonideal behaviors.
In particular, we analyze the spectral properties of the adopted ternary sequences in the presence of nonuniform DACs. This aspect has not been studied in detail by the above mentioned research works, but it is a relevant issue in the domain of instrumentation and measurement.
% may be affected by DAC nonidealities. 

\section{Synthesis of Ternary Sequences}
\label{s:intro_ternary}
In this section, we focus our attention on the spectral properties of practical ternary pseudorandom sequences.
We start by presenting an analysis of the ideal case, already provided in the literature in \cite{Barker&Godfrey1999} and \cite{Tan}, among others.
We then proceed to study the case where the ternary sequences are generated by a practical DAC with nonuniform levels.

\subsection{Ideal Case}

Following the derivations in \cite{Barker&Godfrey1999} and \cite{Tan},
define $u[n]$ as a ternary sequence taking values in $\{-1,0,1\}$ when $n=1, \ldots, N$, with $N$ being a multiple of $6$.
The condition $U[k]=\sum_{i=1}^Nu[i]
\exp\left( -\frac{j2\pi n k}{N}\right)$, 
being zero at even frequencies and at multiplies of $3m$, $m$ integer, requires:
\begin{equation}
	u[n]\!+\!u\left[n\!+\!N/2\right]\!=\!0 ; \, u[n]\!+\!u\left[n\!+\!N/3\right] \!+\! u\left[n\!+\!2N/3\right]\!=\!0
\label{prop}
\end{equation}

In the direct synthesis technique described in \cite{Tan}, 
$u[\cdot]$ is obtained by multiplying a basic  binary sequence  
$u_{basic}[\cdot]$ taking values in $(-1,1)$ 
by the special sequence $[1\; 1\; 0\; -1\; -1\; 0]$. By periodizing the product sequence, 
the $0$'s in $u[n]$ are exactly located at $n=3m$. 

If the signal is generated by an ideal DAC with uniform levels, the condition \eqref{prop} ensures that harmonic multiples of two and three are entirely suppressed.
However, if the signal is generated by a real DAC with non-uniform levels, the suppression is not perfect and undesired harmonic components are present, as we show in the following subsection.

\subsection{Effect of DAC non-uniform levels on the spectrum of ternary sequences} 
\label{sec:nonuniform}

In order to quantify the effect of non-ideal levels in the DAC, the sequence $u[\cdot]$ is assumed as being generated by a nonuniform DAC that maps the input values $(-1,0,1)$
to the output voltages $a_{-1}, a_0, a_1$, with the only constraint $a_{-1}<a_0<a_1$. 
Thus the DAC output sequence $y_{DAC}[\cdot]$ is equal to
\[
	y_{DAC}[n] = \sum_{q=-1}^1 \left[u[n]=q\right]a_q
\]
where the notation 
$[f(x)=q]$ represents the indicator function of the event $\left\{ x \, \mathrm{such} \, \mathrm{that} \, f(x)=q \right\}$.

By defining $\beta = \frac{a_{-1}+a_1}{2}$, 
%the new voltages $\overline{a}_k = a_k-\beta$, $k=-1,0,1$ 
%can be defined. Moreover, define $\alpha=\overline{a}_1$ then
and $\alpha=a_1-\beta$, then
\[
	y_{DAC}[n] = \alpha u[n] + \beta + (a_0-\beta) \left[u[n]=0\right]
\]
which shows that apart from a constant gain $\alpha$ and a constant offset $\beta$,
$y_{DAC}[\cdot]$ differs from $u[\cdot]$ by a nonlinear error sequence $e_{DAC}[n]=(a_0-\beta) \left[u[n]=0\right]$.
Observe that when the DAC is uniform $a_{-1}=-a_1=-1$ and $a_0=0$, so that $e_{DAC}[\cdot]$ is identically zero.
The problem of spectral distortion due to 
nonuniform DACs is conventionally solved by resorting to {\em 
dynamic element matching} techniques or by using {\em mismatch shaping} DACs,
for example in the 
case of $\Sigma\Delta$ ADCs  \cite{Galton&Carbone1995}.
% However, the synthesis of $u(\cdot)$ through off--the--shelf waveform synthesizers prevents from using compensating methods at the very large scale integration level.
Such techniques can be applied in the design phase of a very large scale integration (VLSI) DAC circuit. However, when employing an existing off--the--shelf waveform generator, they are clearly not applicable.

Next it will be shown how the nonlinear error contribution due to 
$e_{DAC}[\cdot]$
destroys the properties (\ref{prop}) and results
in frequency components not present in $U[\cdot]$.
In fact, the error sequence $e_{DAC}[\cdot]$ produces a spectral contribution that can be evaluated as:
\begin{equation}
	E_{DAC}[k] =\sum_{n=1}^N e_{DAC}[n]\exp\left( -\frac{j2\pi n k}{N}\right)
	\label{DACspectrum}
\end{equation}
Since $e_{DAC}[n]$ is identically zero when $n \neq 3m$, $m$ integer, (\ref{DACspectrum}) becomes
\[
	E_{DAC}[k] = (a_0-\beta) \sum_{n=1}^{N/3} \exp\left( -\frac{j2\pi 3 n k}{N}\right)
\]
Given the orthogonality of the complex exponential basis functions
the only nonnegative 
values of $k$ that result in $E_{DAC}[k]\neq 0$ are $k=0, \nicefrac{N}{3}$ for which
$E_{DAC}[k]= (a_0-\beta) \frac{N}{3}$.  
Thus the effect of nonuniform DAC output levels results in a nonlinear 
contribution that adds a mean value and a harmonic component at an {\em even} frequency $k=\frac{N}{3}$ 
that should ideally be zero.
 
A simulation was run with $a_{-1}=-1.3, 
a_0=0.15, a_1=1.0$ and $N=42$. Spectra associated to the 
synthesized signal are shown in Fig.~\ref{genspec}. 
In this figure blue circles represent 
the spectrum generated by the nonuniform DAC, red crosses show the nonlinear error spectrum, and black dots denote the spectrum synthesized by the ideal sequence $u[\cdot]$. The contributions at DC and at $k=N/3=14$ are clearly visible.

\begin{figure}
\centering
\includegraphics[width = 0.95 \columnwidth]{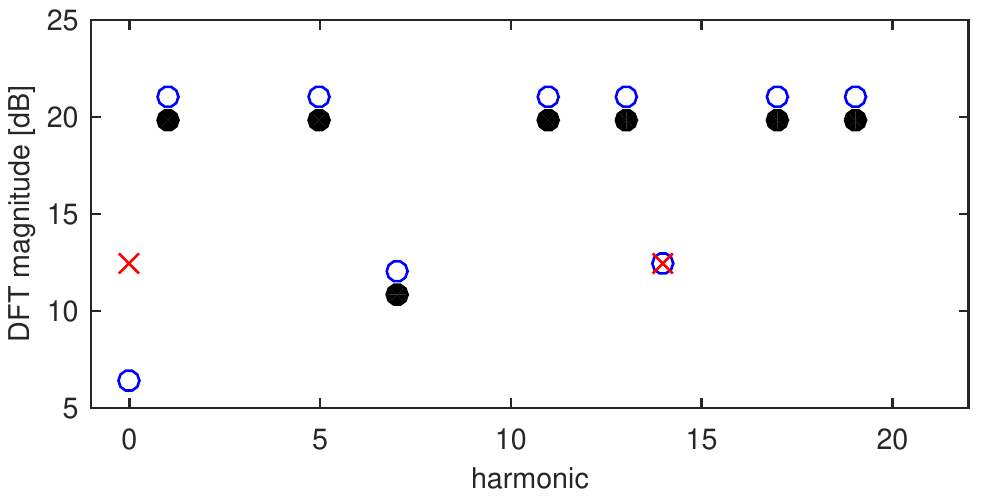}
\caption{Spectra associated to the synthesis of a ternary signal using a nonuniform DAC, with $a_{-1}=-1.3, 
a_0=0.15, a_1=1.0$ and $N=42$:  
uniform DAC -- ideal behavior (black dots); nonlinear error (red crosses); nonuniform DAC (blue circles). \label{genspec}}
\end{figure}

\section{Other implementation issues}
\label{s:implementation_issues}

In a practical implementation, DAC non-uniformity is not the only aspect that influences the properties of the synthesized sequences. 
In particular, another implementation issue is related to the duration of each sequence element (denoted as \emph{chip}). For theoretical results to be applicable, in fact, the uniformity of such duration should be ensured.
Non-uniformity may arise due to the principle of operation of Direct Digital Synthesis (DDS) instrumentation \cite{Agilent33220ADatasheet}. 
% Specifically, when there is a non-integer ratio between waveform memory depth and sequence length, the DDS device performs waveform ``stretching'' to fill the internal waveform memory. 
Specifically, when waveform memory depth is larger than 
sequence length $N$, the DDS device performs waveform ``stretching'' 
to fill the internal waveform memory, which is inaccurate when memory 
depth is not an integer multiple of $N$.
It can be mitigated by acquiring only one sample per chip, or by repeating the sequence multiple times for a more effective memory usage.  

Furthermore, synchronization between the signal generation device and the acquisition device should be ensured. If synchronization is not possible, the effect of leakage can be mitigated by the methods in \cite{SchoukensEtAl2003}.

Moreover, when generating a continuous-time signal from the digital sequence using the zero-order hold approach, the discrete spectrum is multiplied by a sinc function, as discussed in \cite{Barker&Godfrey1999}.
This effect should be taken into account when designing the acquisition system. 

Since the ternary signal contains sharp transitions, jitter of the sampling frequency should be taken into account. 
In fact, the higher the maximum slope of the signal, the more sensititive to jitter the acquisition system becomes.
To mitigate the effect of the sharp edges, a hardware low-pass filter should be used, providing also the benefit of reducing overshoot. 

Fundamental limitations are given by noise in the acquisition system. Wideband additive noise can be mitigated by coherent averaging over multiple periods.
Observe also that the finite resolution of the DAC and ADC used for generating and acquiring the sequence is a fundamental limitation for the attainable spurious free dynamic range (SFDR).
In the remainder of this paper, SFDR is defined as the ratio between the highest desired harmonic and the highest undesired harmonic. 
Furthermore, a meaningful indicator of signal quality is the total harmonic distortion (THD), given by the ratio between the power of the undesired harmonic components and that of the desired harmonic components \cite{IEEE1241}.
%Furthermore, we define the signal-to-noise-and-distortion ratio (SINAD), as the ratio of the root-mean-square (RMS) amplitude of the signal to the RMS amplitude of the noise and distortion components \cite{IEEE1241}, i.e.
%$\text{SINAD} \triangleq 20 \log_{10}{\frac{V_{SI}}{V_{NAD}}}\,,$ 
%where $V_{SI}$ denotes the RMS value of the desired harmonics, i.e. $V_{SI}=\frac{1}{N}\sqrt{\sum X^2_{SI}}\,,$ and $V_{NAD}$ that of all other components discarding DC, i.e. $V_{NAD}=\frac{1}{\sqrt{N\left(N-N_{SI}-1\right)}}\sqrt{\sum X^2_{NAD}}$. 
%Here, $X_{SI}$ denotes the vector containing the spectral magnitude values of the desired harmonics and $X_{NAD}$ that of the undesired components discarding DC, $N$ is the total number of samples, and $N_{SI}$ is the number of desired components.

\section{Randomized constrained sequences with harmonic multiples of two and three suppressed}
\label{s:proposed_method}
\subsection{Construction of the sequences}
\label{s:construction}
The two constraints in (\ref{prop}) may suggest that it is possible to generate sequences of length that is a 
multiple of $6$ and that at the same time suppress harmonic multiples of two and three.
This can be done by starting from the vector $b_0 =[-1\; 0\; 1]$ and constructing the $3$ subsequences
$r_1$, $r_2$ and $r_3$ as follows:
\begin{equation} \label{eq:subsequences}
r_1 =
\left[
\begin{array}{c}
r_{11} \\
r_{12} \\
\vdots \\
\vdots \\
r_{1\frac{N}{6}}
\end{array}
\right]
\quad
r_2 =
\left[
\begin{array}{c}
r_{21} \\
r_{22} \\
\vdots \\
\vdots \\
r_{2\frac{N}{6}}
\end{array}
\right]
\quad
r_3 =
\begin{rcases}
\left[
\begin{array}{c}
r_{31} \\
r_{32} \\
\vdots \\
\vdots \\
r_{3\frac{N}{6}}
\end{array}
\right]
\end{rcases}
\text{$\frac{N}{6}$}
\end{equation}
such that for every $i=1, \ldots, \frac{N}{6}$ the $\left( 1 \times 3 \right)$ vector
$
	[r_{1i} \; r_{2i} \; r_{3i}]
$
is obtained by permuting, at random, the elements in the vector $b_0$.
Then the final sequence is obtained by juxtaposition of the vectors $r_1, r_2, r_3$ as follows:
\begin{align}
x=[
\begin{array}{cccccc}
	r_1^{T} & -r_2^{T} & r_3^{T} &  -r_1^{T} &  r_2^{T} & -r_3^{T} 
\end{array}
]
\label{eq:juxtaposition}
\end{align}
where the superscript $T$ denotes the transpose operation. 

The sequence $x$ has length $N$ and, by construction, its elements satisfy both constraints in (\ref{prop}).
Observe also that by swapping any two values in different subsequences, e.g. in $r_1$
 and $r_3$, the harmonic suppression properties still hold. Swapping can be applied to eventually improve the performance of
 the generated sequence against given criteria, e.g. RMS value or standard deviation of amplitudes 
 in the generated harmonics (see \cite{Tan}). Thus, once a performance 
 criterion is established, both swapping and new sequences can be generated to improve the obtained performance.
 
In the following, any sequence synthesized using the alternative generation approach proposed here will be referred to as \emph{Randomized Constrained Sequence} (RCS).
Conversely, any sequence generated using the direct generation approach in \cite{Tan} will be denoted as \emph{Direct Sequence} (DS).
The operation of the proposed generation method is summarized in Algorithm 1.

\begin{algorithm}[t] \label{alg1}
\caption{Generate a ternary sequence with harmonic multiples of two and three suppressed. The parameters $K_{max}$ and $J_{max}$ are selected by the user so to limit processing time.}\label{alg1}
\begin{algorithmic}[1]
\Procedure{Generate RCS $x$}{}
\State $ N \gets \text{length of sequence}$
\State $ b_0 \gets \left[ -1 \, 0 \, 1 \right]$
\State Define performance criterion $ V(x)$
\State $ V_{min} \gets V_0$
\State $ x \gets [\;] $
\For{$k=1$ to $K_{max}$}
	\State Randomly construct $r_1$, $r_2$, $r_3$, as in \eqref{eq:subsequences}
	\State 
	$
	x_k \gets \left[
	\begin{array}{cccccc}
		r_1^{T} & -r_2^{T} & r_3^{T} &  -r_1^{T} &  r_2^{T} & -r_3^{T} 
	\end{array}
	\right]
	$
	\If{$V(x_k) < V_{min}$}
		\State $V_{min} \gets V(x_k)$
		\State $ x \gets x_k $
	\EndIf
\EndFor

\For{$j=1$ to $J_{max}$}
	\State $x_j \gets$ Randomly swap two values in $r_1$, $r_2$, $r_3$
	\If{$V(x_j) < V_{min}$}
		\State $V_{min} \gets V(x_j)$
		\State $ x \gets x_j $
	\EndIf
\EndFor

\State \Return $x$

\EndProcedure
\end{algorithmic}
\end{algorithm}

 With respect to the DS, this approach:
 \begin{itemize}
 \item does not have a low-amplitude component in its harmonic content, unlike the DS (see figures 2 and 3);
 \item does not concentrate the effect of non--ideal DAC levels in two single components (at $0$ frequency and at $N/3$)
 but spreads the non--ideal error over a larger set of harmonics resulting in a better range free of unwanted harmonic contributions (see Fig. 3);
 \item allows for the generation of several different sequences possibly used to source multiple--input single--output system while few uncorrelated direct sequences are available (stated in \cite{Tan});
 \end{itemize}
However, for long sequences, RCSs have a somewhat larger spread in the amplitudes of the generated harmonics.
In fact, the DS has all desired harmonics of the same amplitude except for one, by construction. On the other hand, RCS presents random fluctuations in all desired harmonics, as illustrated by numerical simulations in the following section.
 
\subsection{Illustration in simulations}
 In Fig. \ref{comp}, the spectra of the DS is plotted (red crosses) together with that of the RCS (black circles), assuming ideal DAC levels and $N=762$.
 Both sequences have components where expected only. 
 The DS shows a smaller standard deviation
 in the amplitude of the generated harmonics. 
 However, there is one frequency in the DS for which the amplitude is rather small. 
 Thus the largest deviation from the maximum amplitude is larger in the DS than in the RCS.
 
When the DAC levels are assumed non ideal ($-1.0005$, $0.001$, $1.001$) the same spectra appear as in Fig. \ref{compnonid}.
 There are the additional expected components at $0$ and $N/3$ while the black spectrum shows an overall increase in 
 the wideband noise but a smaller component at $N/3$ and the same component at $0$ frequency (points overlap).
 In Fig. \ref{compnonid} the red harmonic at $n=127$ is wanted but has a rather small amplitude, while the component at $N/3$ is unwanted and has a rather
 large amplitude.
Hence a reduction by about 20 dB of the largest unwanted component is obtained with the proposed method.
\begin{figure}
\centering
\includegraphics[width=0.95\columnwidth]{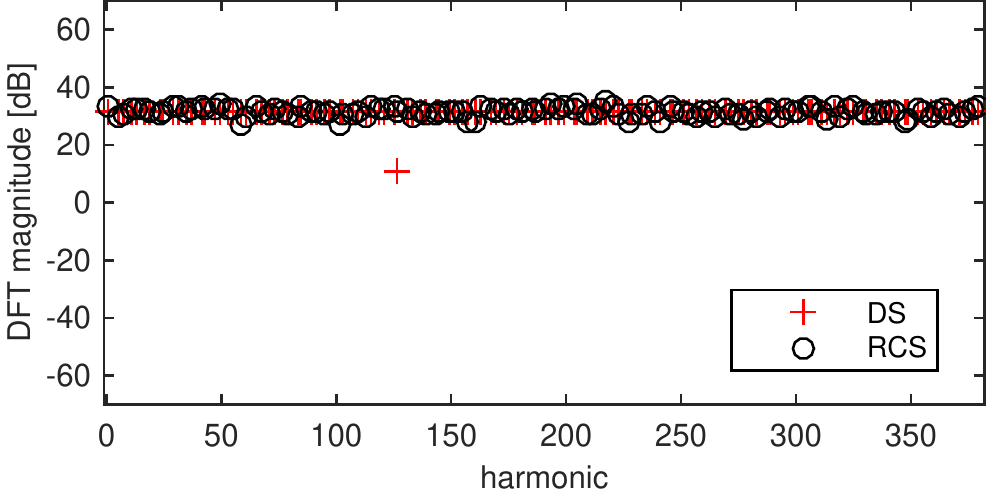}
\caption{Simulation results. Spectra associated to the synthesis of one period of a ternary sequence of length 762 using the DS method (red crosses) and the RCS approach described here (black circles), with an ideal DAC. \label{comp}}
\end{figure}  

\begin{figure}
\centering
\includegraphics[width=0.95\columnwidth]{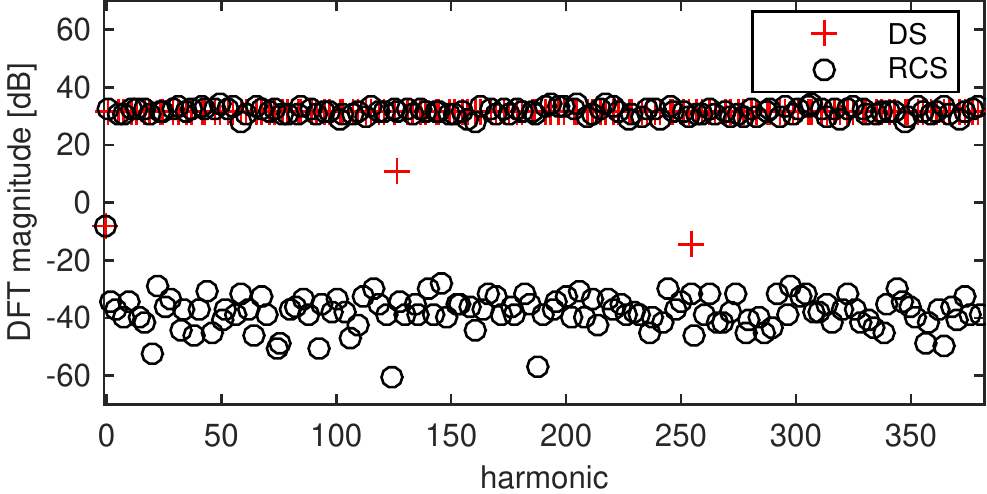}
\caption{Simulation results. Spectra associated to the synthesis of one period of a ternary sequence of length 762 using the DS method (red crosses) and the RCS approach described here (black circles), with a non--ideal DAC. Effect of non--ideal DAC levels: a large harmonic appears at $N/3$ in the red spectrum. \label{compnonid}}
\end{figure}

\subsection{Theoretical Characterization}
In this section, we provide two theorems that completely characterize the frequency domain behavior of the RCS, both in the ideal DAC case and in the non-uniform DAC case. 
\begin{theorem}
Consider the RCS $u_i$, $i=1 , \dots , N$, constructed according to the procedure described in Section~\ref{s:construction}. 
Its power spectral density assuming an ideal DAC is given by 
\begin{align}
R_U[k] = \frac{2}{3} \left( 1 - (-1)^k + \cos{\left(\frac{\pi}{3} k \right)}
	- \cos{\left(\frac{2}{3} \pi k\right)} \right) \,.
\end{align}
\end{theorem}
\begin{proof}
See Appendix A.
\end{proof}

\begin{theorem}
In the presence of DAC non-uniform levels as defined in Section~\ref{sec:nonuniform}, the power spectral density of the RCS is given by
\begin{align}
\label{eq:theorem2}
R_Y[k] = N\left(\frac{1}{3}a_0+\frac{2}{3}\beta\right)^2 \delta[k]+ \alpha^2R_U[k] + R_V[k]
\end{align}
where $R_U[k]$ is the power spectral density of the RCS in the ideal case as defined in Theorem~1,
$ \delta[k]$ denotes the Kronecker delta function and
\begin{align}
\begin{split}
	R_V[k] & = \frac{2}{3}
	 (a_0-\beta)^2 \left(\delta[k \bmod 2]-\delta[k \bmod 6] \right)
\end{split} 
\end{align}
\end{theorem}

\begin{proof}
See Appendix B.
\end{proof}

\begin{figure}
\centering
\subfigure
{\includegraphics[width = 0.95 \columnwidth]{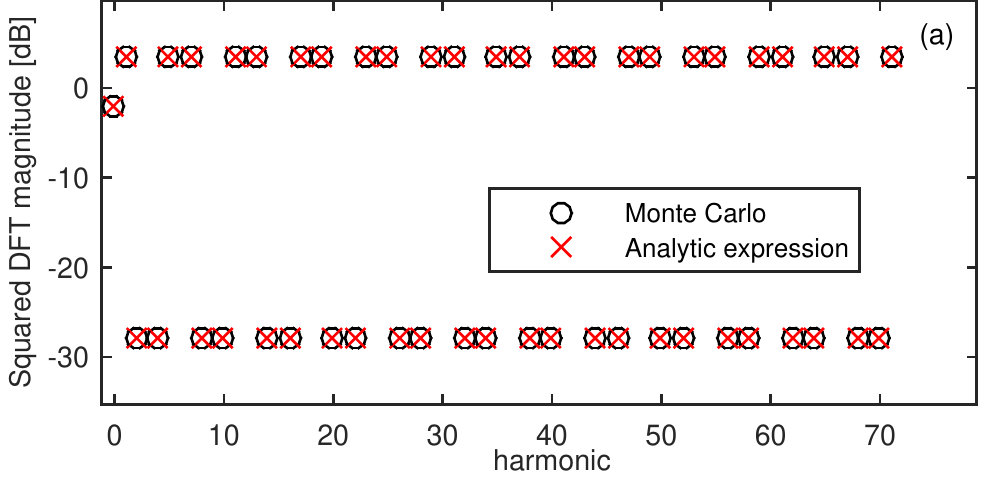}}
\vskip -5pt
\subfigure
{\includegraphics[width = 0.95 \columnwidth]{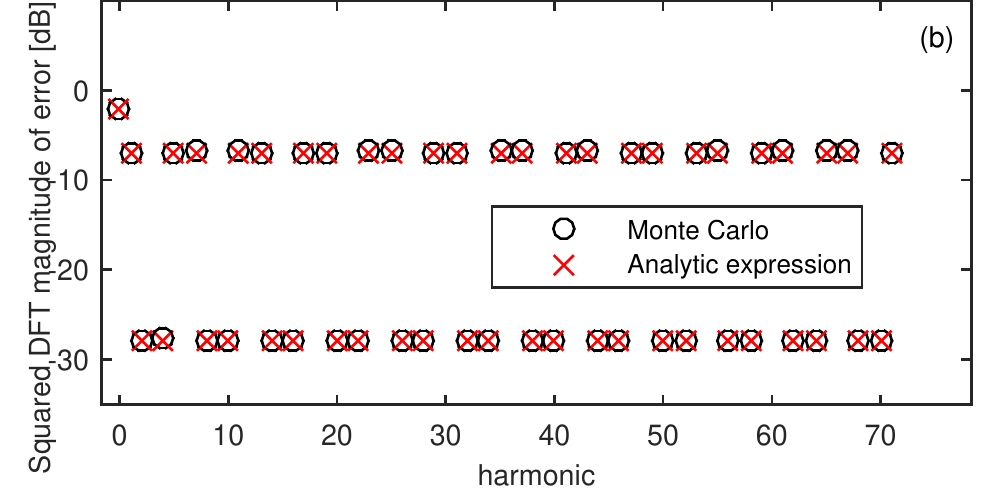}}
\caption{
Simulated spectrum of RCS of length 144 using a non-uniform DAC with levels $\left( -1 ,\, 0.1 ,\, 1.1 \right)$. Comparison between Monte Carlo simulation and analytic expression in \eqref{eq:theorem2}: (a) full spectrum; (b) spectrum of the error due to DAC non-uniformity. Monte Carlo results were obtained by averaging the squared DFT magnitude over $10^6$ realizations.  
Here, we defined the squared DFT magnitude $A^2$ in deciBel as $10 \log_{10}(A^2) = 20 \log_{10}(A)$.
Analytic results were obtained by multiplying the power spectral density \eqref{eq:theorem2} by $N$.
\label{fig:MonteCarlo}
}
\end{figure}

\begin{figure}
\centering
{\includegraphics[width = 0.95 \columnwidth]{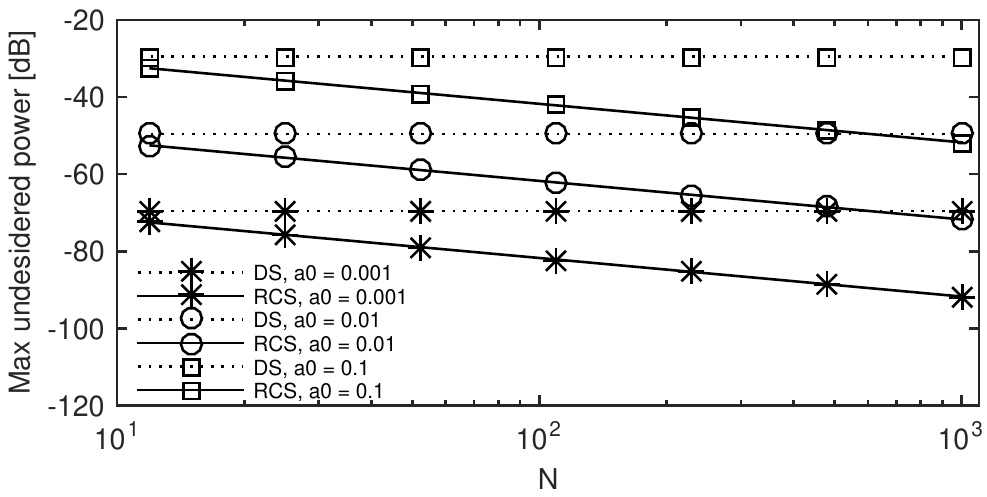}
\caption{
Theoretical behavior of maximum undesired harmonic power level vs sequence length $N$, for different DAC errors. The DAC levels are $\left( -1 ,\, a_0 ,\, 1 \right)$. Harmonic amplitude levels are normalized with respect to $N$. 
Specifically, DS results are obtained from $(E_{DAC}[k])^2/N^2$; RCS results are obtained from $R_U[k]/N$.
}
\label{fig:maxharm_vs_N}
}
\end{figure}

Theorem~2 shows that when DAC levels are non-uniform, the RCS is characterized 
by an error spectrum containing a DC and a wide-band spectral content, having components also at even
harmonics.  
Numerical results in Fig.~\ref{fig:MonteCarlo}(a) and \ref{fig:MonteCarlo}(b) provide a comparison between \eqref{eq:theorem2} and a Monte Carlo simulation considering a non-uniform DAC. A good match between theoretical derivations and numerical results is demonstrated. 

Moreover, we can observe from \eqref{eq:theorem2} that the power of the largest undesired harmonic is $\frac{2}{3}(a_0-\beta)^2$. On the other hand, based on the derivations in Section~\ref{sec:nonuniform}, the power of the highest undesired harmonic generated by the DS can be derived as $\frac{N}{9}(a_0-\beta)^2$.
Therefore, the randomization provided by the RCS results in a reduction of the power of the largest undesired harmonic by a factor of $\frac{N}{6}$. 

The theoretical behavior of the maximum undesired harmonic level for the DS and RCS sequences is shown in Fig.~\ref{fig:maxharm_vs_N} as a function of $N$ for several different DAC error levels. It is possible to notice that RCS always provides better performance than DS for the same DAC error, since $N>6$ for all practical applications.

\section{Experimental results}

The theoretical results described in Sections \ref{s:intro_ternary} and \ref{s:proposed_method} have been validated experimentally. 
Ternary sequences were synthesized offline using the DS and RCS approaches. These sequences were stored in the memory of a signal generator device, which was then used to generate voltage waveforms based on the stored sequences. 
The following subsections describe two versions of the employed experimental setup and related results. 

\subsection{Experimental Setup 1 - DDS AWG system}

The first set of experiments was conducted using the setup shown in Fig. \ref{fig:block_diagram}. 
A 14-bit Arbitrary Waveform Generator (AWG), 33220A by Agilent, operating according to the DDS principle, was used to generate a designed DS of length 42, according to \cite{Tan}, which was loaded in the arbitrary waveform memory. 
The repetition period of the sequence was set to 10 ms. 
The signal generated by the AWG was directly connected via a coaxial cable to an acquisition board (12-bit PCI 5105 by National Instruments) and synchronously sampled at a rate of 10 MSa/s. A total of $4\cdot 10^{6}$ samples were acquired, corresponding to 40 periods of the generated sequence.

\begin{figure}
\centering
\includegraphics[width = 0.95 \columnwidth]{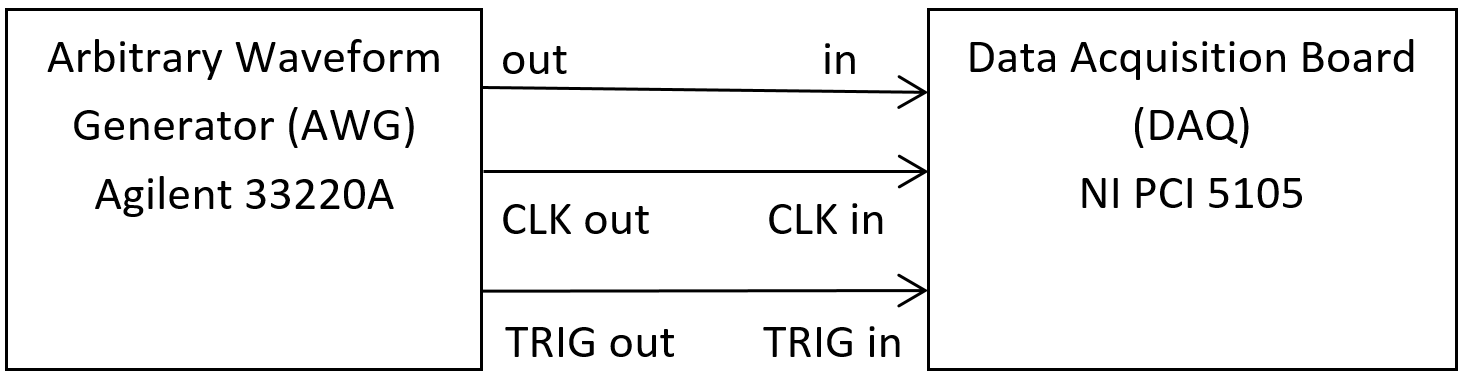}
\vskip 5pt
\caption{Diagram of the experimental setup used for the characterization of the AWG generating the ternary sequence. \label{fig:block_diagram}}
\end{figure} 

Results using such a setup are shown in Fig. \ref{fig:AWG}. In this case the ratio between waveform memory depth (65536) and sequence length (42) is not an integer. Results show that, when loading only one repetition of the sequence in memory, the automatic stretching operated by the instrument causes non-uniform duration of the sequence elements (\emph{chip duration}). Then, this non-uniform chip duration results in undesired harmonic content. 
This can be observed in Fig. \ref{fig:AWG}(a), where undesired components at 300, 900, 1500, and 2100 Hz are higher than the other undesired harmonics by about 30 dB. 
Conversely, using a repeated sequence to fill the waveform memory, this effect is mitigated and an SFDR of 80~dB is obtained, as shown in Fig. \ref{fig:AWG}(b) where the effect of the DC component is neglected. 

\begin{figure}
\centering
\subfigure
{\includegraphics[width = 0.95 \columnwidth]{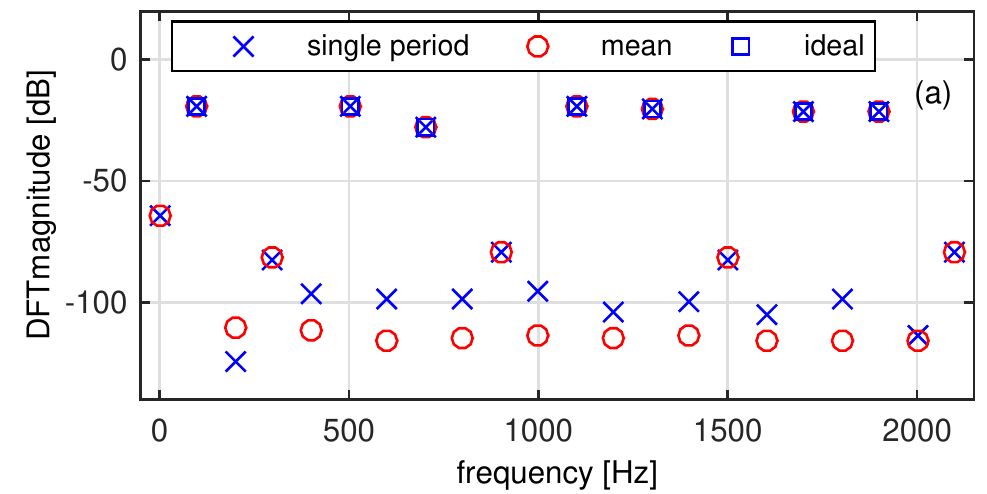}}
\subfigure
{\includegraphics[width = 0.95 \columnwidth]{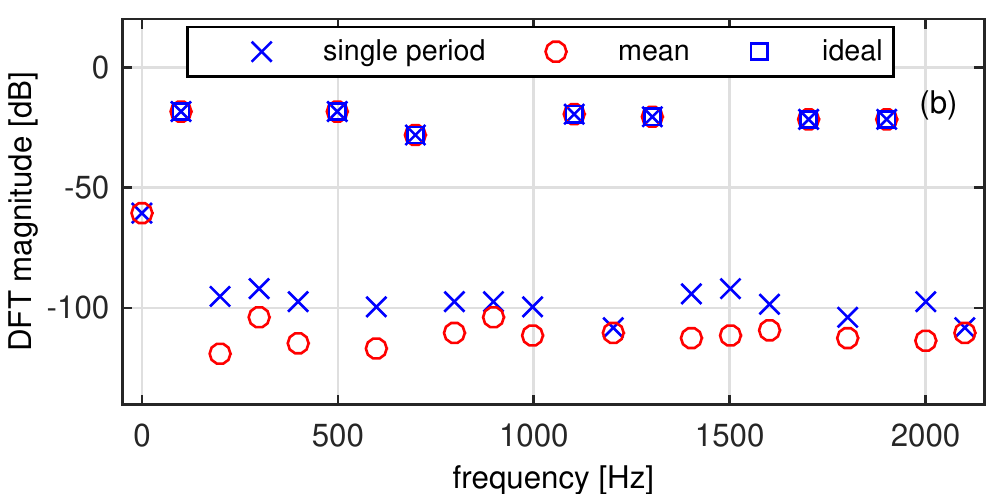}}
\caption{Estimated spectra obtained with the DS using the setup in Fig. \ref{fig:block_diagram}: (a) single repetition of the sequence in the AWG memory; (b) multiple repetition to fill the AWG memory (sequence repeated a fractional number of times). The mean spectrum is obtained by coherent averaging over 40 periods. \label{fig:AWG}}
\end{figure} 

This experimental setup can potentially operate at frequencies up to 20 MHz. However, it is prone to some of the implementation issues described in Section \ref{s:implementation_issues}, mainly non-uniform duration of elements within the sequence and limited ADC resolution. Therefore, a different experimental setup that addresses such issues, albeit operating only up to audio frequencies, was employed for further validation. This setup is described in the following subsection.

\subsection{Experimental Setup 2 - PC sound card}

\subsubsection{Setup}
Periodic sequences of length 42 were generated by means of the DAC in the sound card (ALC887 by Realtek) of a personal computer. The acquisition was performed using the ADC on the same sound card. The DAC and ADC were set at 42 kSa/s sample rate, with a resolution of 16 bits.
The sequence period was 10 ms, its repetition rate was set to $10^2$ Hz, resulting in a chip rate of 4.2 kHz. 
Coherent sampling was performed, with the DAC and ADC using the same internal clock source. 
The duration of the acquisition was $10^2$ s, corresponding to $10^4$ periods. Initial DAC and ADC transients of duration 0.1 s were discarded.

\subsubsection{Direct connection results}
Results are presented in Fig. \ref{fig:PC_42}(a) and \ref{fig:PC_42}(b), for the DS generated according to \cite{Tan} and for the RCS proposed in Section \ref{s:proposed_method}, respectively. 
The DAC output and the ADC input of the soundcard were connected together using a cable. 
The discrepancy between the ideal levels of the desired harmonics and the actual measured levels is due to a scale factor related to the amplitude range of the ADC.
Such scale factor can be compensated by calibration.
From Fig. \ref{fig:PC_42}, it is possible to observe that, by averaging over $10^4$ periods, an SFDR of approximately 100 dB
is obtained using both methods. Furthermore, calculations show that both methods achieve a THD of -104 dB.
We also stress that it is not possible to discriminate whether the undesired harmonics are caused by the DAC or by the ADC.
The DC component is excluded from the calculation of the SFDR because it can be relatively easily attenuated by means of a series capacitor or 
by averaging and subtraction operations.
Moreover, it can be observed that the difference between the highest desired harmonic and the lowest desired harmonic is approximately 5 dB for the RCS and 9 dB for the DS. 

\begin{figure}
\centering
\subfigure
{\includegraphics[width=0.95\columnwidth]{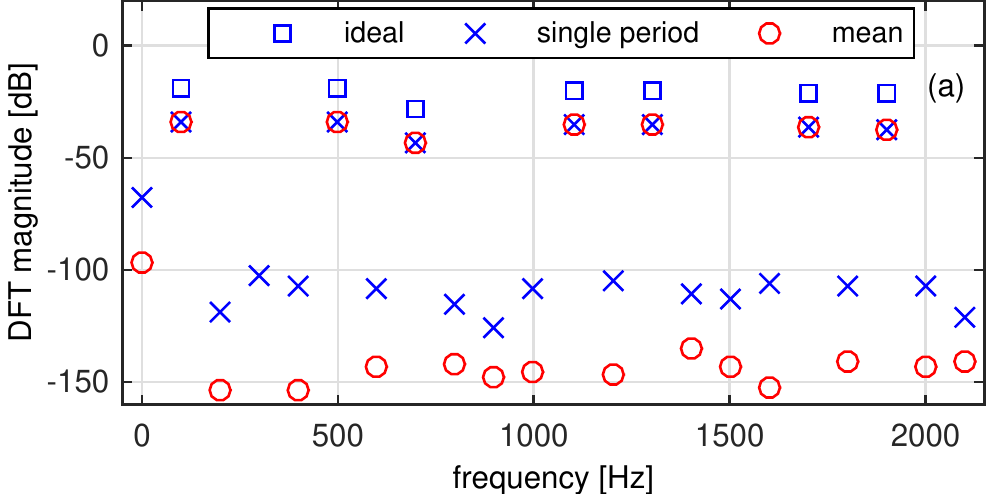}}
\subfigure
{\includegraphics[width=0.95\columnwidth]{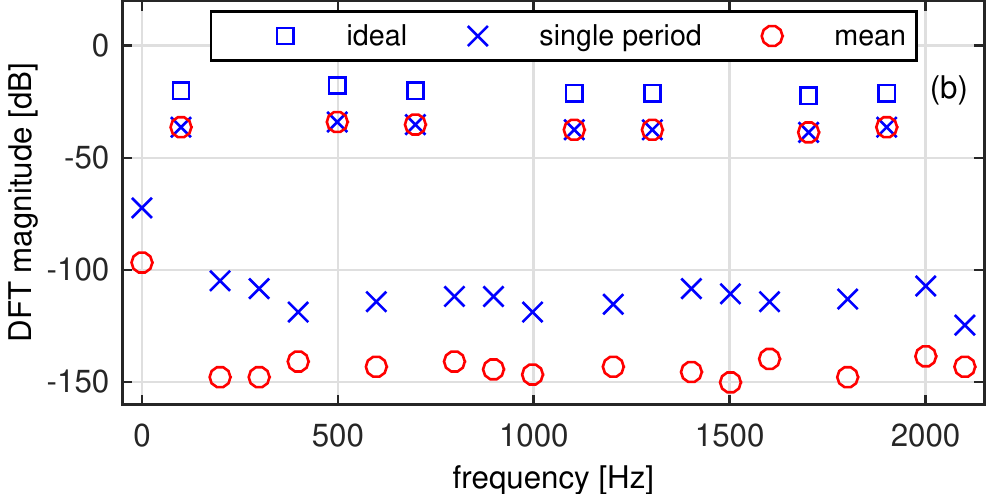}}
\caption{Estimated spectra related to the PC sound card acquisition of: (a) the DS in \cite{Tan}; (b) the proposed RCS. The mean spectrum is obtained by coherent averaging over $10^4$ periods.}
\label{fig:PC_42}
\end{figure}

\subsubsection{DAC non-uniformity emulation}
A significant DAC non-uniformity was emulated by substituting a value of $10^{-3}$ to the zero values, both in the DS and in the RCS. 
Acquisition results for the DS are shown in Fig.  \ref{fig:PC_42_nonunif}(a).  
The component at 1400 Hz, corresponding to the harmonic number $N/3=14$ is considerably higher,
resulting in an SFDR of 53 dB and a THD of -59 dB.
This validates the theoretical derivations in Section \ref{sec:nonuniform}, where it is shown that non-uniformity in the DAC levels causes the power of the harmonic at $N/3$ to increase when using the DS. 
This phenomenon has a negative effect on performance, because it contradicts the design objective stated in Section \ref{s:intro_ternary}. 
In fact, since the even-order harmonics are not entirely suppressed, it is not possible to eliminate the effect of even-order nonlinearities on the frequency response measurement.
However, the DS may be conveniently used to highlight the presence of DAC non-uniformity, just by observing one specific harmonic.
On the other hand, an acquisition using the RCS is shown in Fig. \ref{fig:PC_42_nonunif}(b). The highest undesired harmonic is at a lower level with respect to Fig. \ref{fig:PC_42_nonunif}(a), namely 58 dB below the highest desired harmonic, illustrating that the proposed method mitigates the effect of DAC non-uniformity, by about 5 dB with respect to DS performance.
This gain will grow as the sequence length $N$ increases.
Therefore, the RCS can be used in those cases where insensitivity to non-uniformity is required, while the DS may be beneficial when the detection of DAC non-uniformity is necessary.

\begin{figure}
\centering
\subfigure
{\includegraphics[width=0.95\columnwidth]{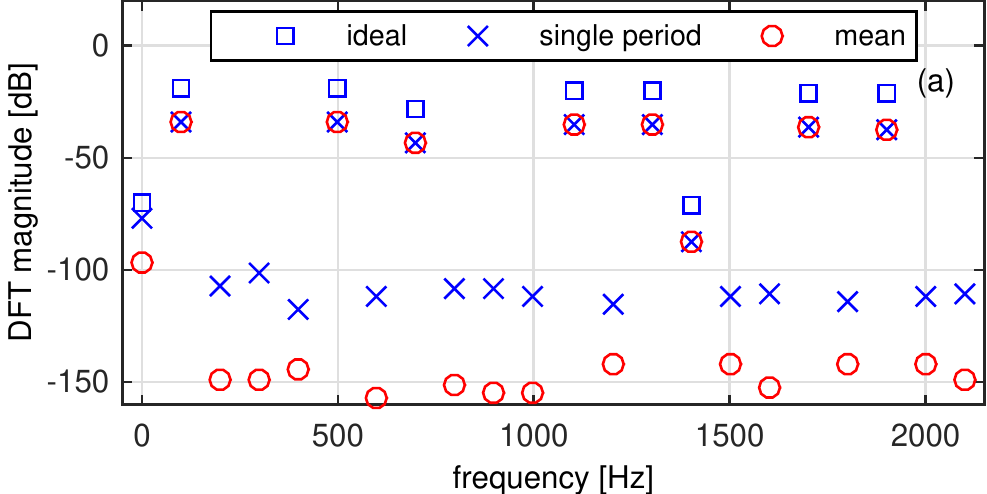}}
\subfigure
{\includegraphics[width=0.95\columnwidth]{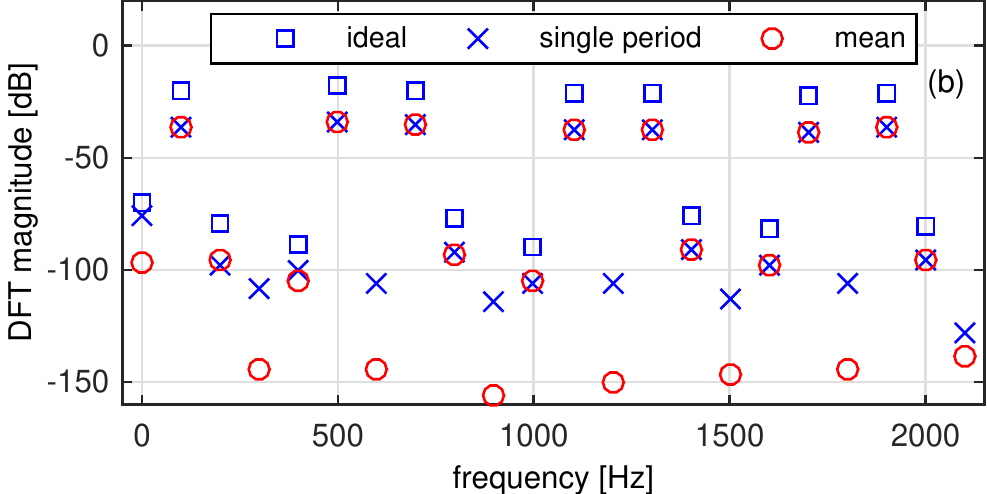}}
\caption{Experimental emulation of DAC non-uniformity. Periodic sequence of length 42 with zeros replaced by $10^{-3}$: (a) DS; (b) RCS. Notice that in (a) the 14th harmonic at frequency 1400 Hz is relatively high, approximately 55 dB higher than the other undesired harmonics.}
\label{fig:PC_42_nonunif}
\end{figure}

\subsubsection{Analog filter effect}
Tests in the presence of a hardware first-order RC low-pass filter at the input of the ADC were conducted. 
The filter is necessary in those applications where sharp edges and non-uniform signal amplitude distributions are problematic, e.g. for ADC testing. 
Without it, in fact, intermediate ADC levels are not excited. 
Results using the DS from this test are shown in Fig. \ref{fig:PC_RC2k}. The filter is realized with $R=1000$ $\Omega$ and $C=80$ nF, yielding a 3-dB cutoff frequency of approximately 2 kHz. 
The spectral properties of the sequence are similar to those of Fig. \ref{fig:PC_42}(a). 
However, the time-domain behavior in Fig. \ref{fig:PC_time_domain} illustrates that the filter smooths the sharp edges and reduces high-frequency noise. 
Similar results for RCS are not presented here for brevity. 
An analysis of the spectrum in Fig. \ref{fig:PC_RC2k} allows the detection of a nonlinear behavior of the system. Specifically, it can be noticed that the 3rd, 9th, and 15th harmonics, at 300 Hz, 900 Hz, and 1500 Hz respectively, are approximately 10 dB higher than the average level of the other undesired harmonics. The presence of this nonlinear component is probably due to an increased loading of the DAC by the filter. An impedance matching circuit could mitigate this behavior.

\begin{figure}
\centering
\includegraphics[width=0.95\columnwidth]{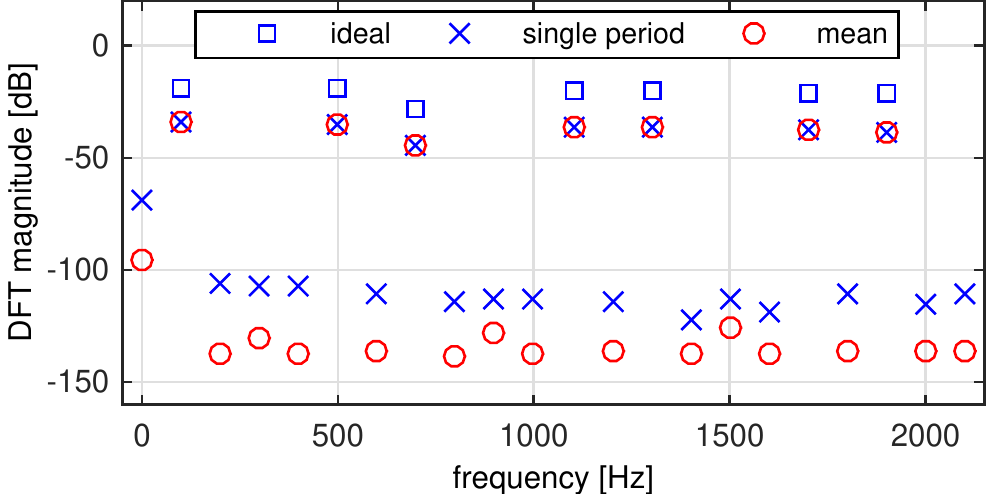}
\caption{Estimated spectra related to an acquisition in the presence of a 2-kHz RC low pass filter.}
\label{fig:PC_RC2k}
\end{figure} 

\begin{figure}
\centering
\subfigure
{\includegraphics[width=0.95\columnwidth]{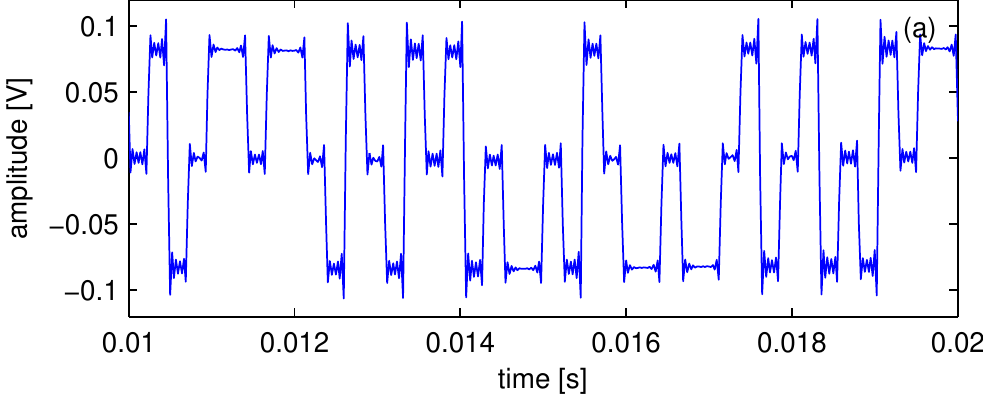}}
\subfigure
{\includegraphics[width=0.95\columnwidth]{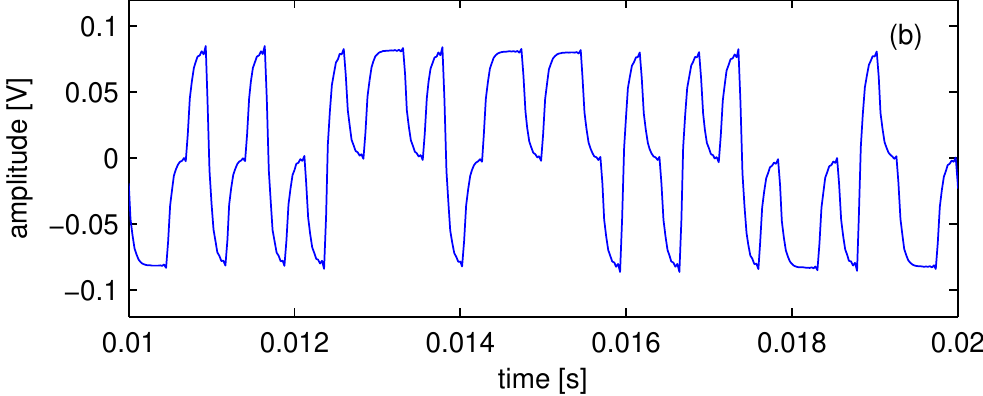}}
\caption{Time-domain behavior of one period of the acquired ternary sequence: (a) direct connection; (b) with hardware RC low pass filter. }
\label{fig:PC_time_domain}
\end{figure}

\section{Extended Experimental Comparison}
In order to provide a more extensive evaluation of the properties of the RCS and DS sequences, we compare them with two further ternary signal designs from the literature. The first is Pseudo-Random Ternary Signals, based on Galois theory \cite{BarkerEtAl2004}, which provides sequences with period $q^{n} - 1$, where $n$ is an integer and $q$ is a prime or a power of a prime. For suppressing harmonic multiples of 2 and 3, it is necessary that $q = 6k + 1$, with $k$ an integer, i.e. q = 7, 13, 19, $\ldots$. These signals have all the non-suppressed harmonics of equal amplitude, which is a desired feature. 
However, their disadvantage in a practical setting is that the number of possible sequence lengths is limited.

The second design considered for comparison is computer-optimized multilevel multi-harmonic sequences \cite{Pintelon&Schoukens2012}. These signals can have zero amplitude at harmonic multiples of 2 and 3 provided the period is an integer multiple of 6, but the amplitude at the desired harmonics is not uniform.

To perform such evaluation, we first designed a pseudorandom sequence of length 48 using the ``Galois'' software, and a pseudorandom sequence of length 42 using the ``multilev'' software\footnote{The ``Galois'' and ``multilev'' software is freely available at \url{https://www2.warwick.ac.uk/fac/sci/eng/research/systems/bbsl/signal_design}}.
Then, we generated and acquired the periodic sequences with the PC sound card experimental setup, using direct connection between the ADC and the DAC. The duration of the single sequence was set to 10 ms, whereas the total duration of the acquisition was 100 s.
Experimental results of the comparison of ``Galois'' and ``multilev'' sequences with the DS and RCS are shown in Fig.~\ref{fig:PC_extended_comparison}.

\begin{figure}[t]
\centering
\subfigure
{\includegraphics[width=0.95\columnwidth]{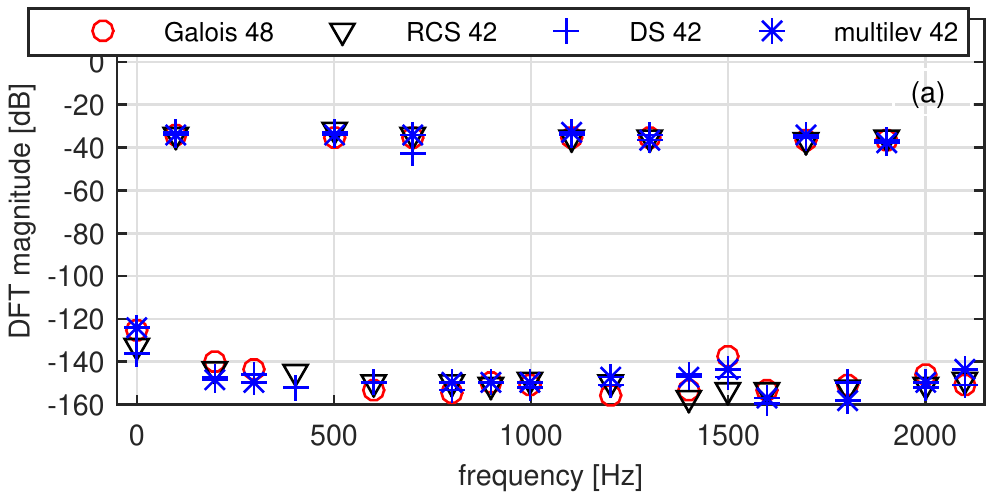}}\\
\subfigure
{\includegraphics[width=0.95\columnwidth]{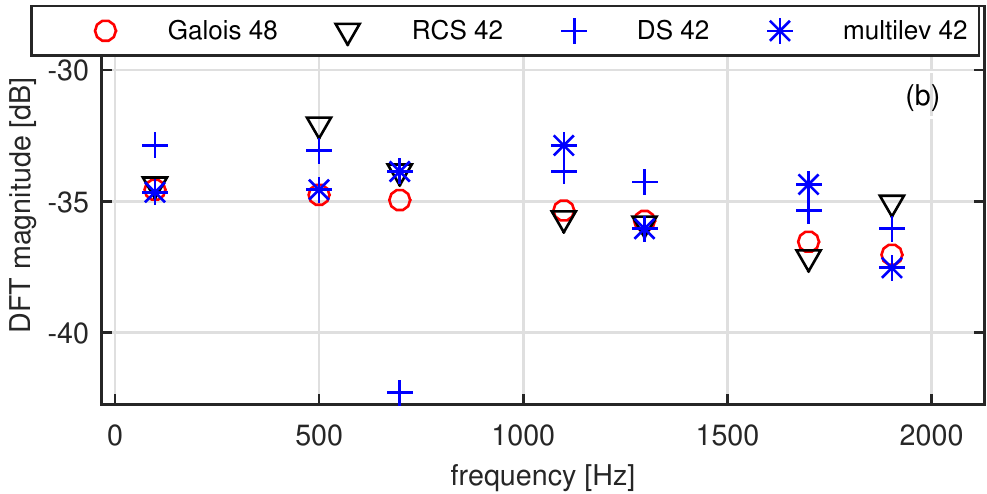}}\\
\subfigure
{\includegraphics[width=0.95\columnwidth]{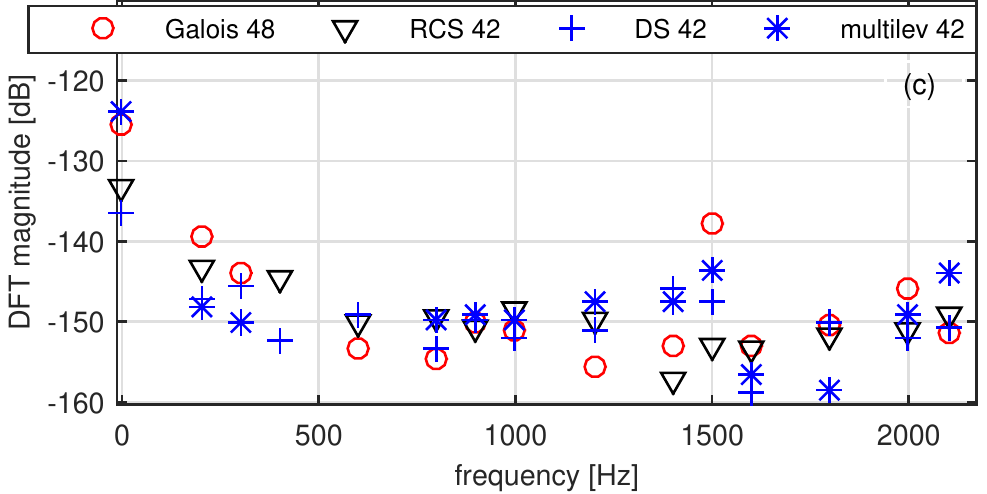}}
\caption{
Extended experimental comparison. Spectra obtained using the PC sound card setup. The proposed RCS is compared with the DS, Galois, and multilev sequences: (a) full picture, (b) magnification of the upper portion, (c) magnification of the lower portion.
}
\label{fig:PC_extended_comparison}
\end{figure}

From Fig.~\ref{fig:PC_extended_comparison}(a), it is possible to notice that all four considered sequences provide an SFDR of approximately 100 dB or larger.
Furthermore, from Fig.~\ref{fig:PC_extended_comparison}(b) it is apparent that the Galois pseudo-random ternary sequence provides the “flattest” response at the desired harmonics. This was expected, since by design such a sequence has all the non-suppressed harmonics of equal amplitude. Conversely, the DS sequence has one of the desired harmonics of amplitude much smaller than the others.
The RCS provides a smaller spread of the excited harmonics with respect to the DS, but a larger spread with respect to the Galois sequence. On the other hand, the Galois sequence has a larger level of the undesired harmonic number 15, see Fig.~\ref{fig:PC_extended_comparison}(c).
Finally, it can be noticed that the performance of RCS and multilev is similar.
Observe that we selected a Galois sequence with a length of 48 while all others have length 42. This is due to the fact that 48 is one of the limited possible lengths that Galois sequences can have. However, the comparison is fair, since the acquisition time is the same (100 s) for all sequences. For the Galois sequence, a sampling rate of 48 kHz was used. 

Results in the presence of a non-uniform DAC are presented in Fig.~\ref{fig:PC_extended_comparison_nonuniform_DAC}.
Such graphs are obtained as in Fig.~\ref{fig:PC_42_nonunif} by artificially replacing the zero values in the sequences with $10^{-3}$.
As shown in Fig.~\ref{fig:PC_extended_comparison_nonuniform_DAC}(b), the DS sequence is considerably affected by the non-uniformity of DAC levels, while the Galois pseudorandom sequence is the least affected. 
The RCS and multilev sequences provide comparable performance, both being robust with respect to DAC non-uniformity.

\begin{figure}[t]
\centering
\subfigure
{\includegraphics[width=0.95\columnwidth]{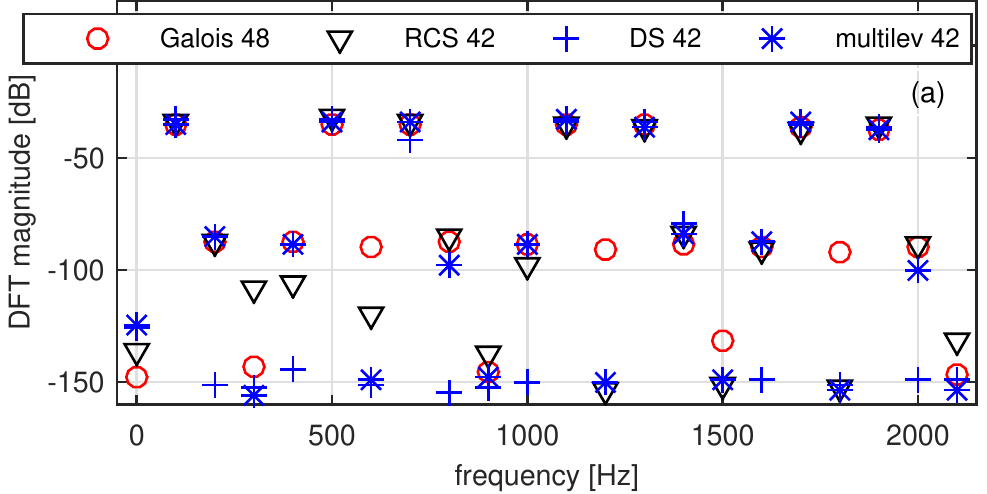}}\\
\subfigure
{\includegraphics[width=0.95\columnwidth]{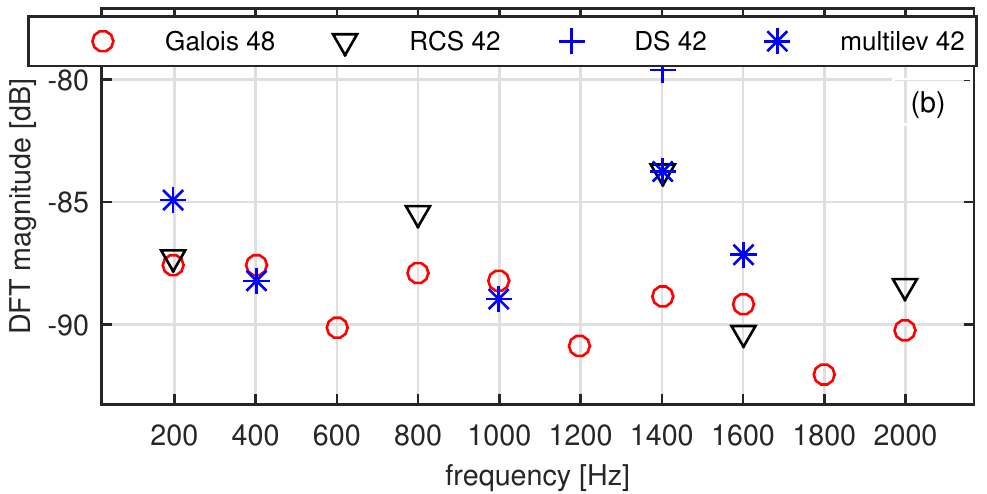}}
\caption{
Extended experimental comparison, with artificially imposed DAC non-uniformity  (levels used: [-1 0.001 1]): (a) full picture, (b) magnification of the lower portion.
}
\label{fig:PC_extended_comparison_nonuniform_DAC}
\end{figure}

\section{Conclusion}
The design of excitation signals for testing dynamic systems was analyzed from a measurement point of view.
The effect of DAC non-uniform levels on the spectral properties of synthesized ternary sequences was analyzed, together with other practical implementation issues, both theoretically and experimentally. 
Furthermore,  a low cost method was developed to generate signals with a very high spectral purity. 
Specifically, a new constrained randomized approach for synthesizing ternary sequences with harmonic multiples of two and three suppressed was proposed, characterized theoretically, and compared experimentally with three approaches from the literature, demonstrating harmonic suppression of the order of 100 dB. 

\appendices
\section{Proof of Theorem 1}
First, let us derive the autocorrelation 
$
r_U[m,l] = E \left( u_m u_l \right) = r_U[n] = E \left( u_m u_{m+n} \right),
$
where we have defined $n=l-m$ and the notation $u_m$ is a shorthand for $u[m]$.
For $n=0$, $r_U[n]=E\left( u_m^2 \right) = \frac{2}{3}$.
Furthermore, for $n=\frac{N}{2}$, we have that $u_{m+N/2}=-u_m$ by construction, thus 
$r_U[n]=-E\left( u_m^2 \right) = -\frac{2}{3}$. 
Moreover, for $n=\frac{N}{6},\,\frac{5}{6}N$, $r_U[n]=\frac{1}{3}$. 
Finally, for $n=\frac{N}{3},\,\frac{2}{3}N$, $r_U[n]=-\frac{1}{3}$. 
In the remaining cases, the random variables $u_m$ and $u_{m+n}$ are independent and zero-mean, thus $r_U[n]=0$.
Therefore, the autocorrelation can be written as
\begin{align}
\label{eq:autocorrelation_ideal}
r_U[n] = \left\{ 
\begin{array}{ll}
\frac{2}{3} & n=0 \\
-\frac{2}{3} & n=\frac{N}{2} \\
\frac{1}{3} & n=\frac{N}{6} \,, \frac{5}{6} N\\
-\frac{1}{3} & n=\frac{N}{3} \,, \frac{2}{3} N\\
0 & \mbox{elsewhere}
\end{array}
\right.  \,.
\end{align}

%% forma con il delta di Kronecker
%\begin{align}
%r_U[n] = \frac{2}{3}\delta\left[n\right] 
%	-\frac{2}{3}\delta\left[n-\frac{N}{2}\right]
%	+\frac{1}{3}\delta\left[n-\frac{N}{6}\right]	
%	+\frac{1}{3}\delta\left[n-\frac{5N}{6}\right]		
%	-\frac{1}{3}\delta\left[n-\frac{N}{3}\right]	
%	-\frac{1}{3}\delta\left[n-\frac{2N}{3}\right]			
%\end{align}

Then, by taking the discrete Fourier transform (DFT) of \eqref{eq:autocorrelation_ideal}, we obtain
\begin{align}
R_U[k] & = 
\frac{2}{3} \left( 1 - \mathrm{e}^{-\mathrm{i} \pi k} \right) \notag \\
& + \frac{1}{3} \left(
 \mathrm{e}^{-\mathrm{j} \frac{\pi}{3}k } 
 + \mathrm{e}^{-\mathrm{j}\frac{5}{3} \pi\, k }
 - \mathrm{e}^{-\mathrm{j}\frac{2}{3} \pi\, k } 
 - \mathrm{e}^{-\mathrm{j}\frac{4}{3} \pi\, k } 
 \right) \notag \\
 & = 
\frac{2}{3} \left( 1 - \mathrm{e}^{-\mathrm{i} \pi k} \right) \notag \\
& + \frac{1}{3} \left(
 \mathrm{e}^{-\mathrm{j} \frac{\pi}{3}k } 
 + \mathrm{e}^{\mathrm{j}\frac{\pi}{3}k }
 - \mathrm{e}^{-\mathrm{j}\frac{2}{3} \pi\, k } 
 - \mathrm{e}^{\mathrm{j}\frac{2}{3} \pi\, k } 
 \right) \,.
\label{eq:exponential}
\end{align}
The theorem follows from \eqref{eq:exponential}, by using Euler's formula.

% The theorem follows from \eqref{eq:exponential}, by using Euler's formula and applying the sum-to-product and product-to-sum prosthaphaeresis formulae.

\section{Proof of Theorem 2}
Denote as $y[\cdot]$ a shorthand of $y_{DAC}[\cdot]$ defined in Section \ref{sec:nonuniform}. We have:
\begin{align}
\begin{split}
	y_my_l = & \left(\alpha u_m+\beta+(a_0-\beta)[u_m=0] \right) \notag \\
	 \times &
	\left(\alpha u_l+\beta+(a_0-\beta)[u_l=0] \right)
\end{split}
\end{align}
By expanding terms, defining $n=l-m$, and taking the expected value, we obtain:
\begin{align}
\begin{split}
r_Y[n] & = E(y_my_l)= \beta^2+\alpha^2 E(u_mu_l)+\alpha \beta E(u_m+u_l) \\
& + (a_0-\beta)\beta E\left( [u_m=0] + [u_l=0]\right)  \\
& + \alpha(a_0-\beta)E\left( u_m[u_l=0]+u_l[u_m=0] \right)\\
& + (a_0-\beta)^2E\left( [u_m=0][u_l=0]\right) \,.
\label{corre}
\end{split}
\end{align}
Since $u_m \in \{ -1,0,1\}$ with equal probability and $[u_m=0]=1$ with probability $\frac{1}{3}$ and $0$ otherwise,
\begin{align}
\begin{split}
	E(u_m) = 0  \quad E([u_m=0]) = \frac{1}{3} \quad m=0, \ldots, N-1 \,.
\end{split}
\end{align}
Now consider the expected value $E\left( u_m[u_l=0]\right)=E\left( u_m[u_{m+n}=0]\right)$ and define
$\mathcal{M} = \{ 0, \frac{1}{6}N,\frac{N}{3}, \frac{1}{2}N, \frac{2}{3}N, \frac{5}{6}N \}$.
If $n \not\in \mathcal{M}$, $u_m$ and $u_{m+n}$ are statistically independent and
\begin{align}
\begin{split}
	E\left( u_m[u_{m+n}=0] \right)= E(u_m)E([u_{m+n}=0]) = 0. 
\end{split}
\end{align}
When $n \in \mathcal{M}$, $u_m$ and $u_{m+n}$ are constrained by \eqref{eq:juxtaposition} and
when $n =  \frac{1}{6}N, \frac{1}{3}N,  \frac{2}{3}N, \frac{5}{6}N$
\begin{align}
\begin{split}
%n = & \frac{1}{6}N, \frac{1}{3}N,  \frac{2}{3}N, \frac{5}{6}N \\
	E\left( u_m[u_{m+n}=0] \right)= & 
	E(E\left( u_m[u_{m+n}=0] | u_{m+n}=0 \right))  \\
	 = & \frac{1}{3}\left( -1\frac{1}{2}+1\frac{1}{2}\right) = 0.
\end{split}
\end{align}
Additionally, when $n=0$
\begin{align}
\begin{split}
%n = & 0\\
%	& 
	E\left( u_m[u_{m}=0] \right)= 0.
\end{split}
\end{align}
Finally, when $n= \frac{1}{2}N$
\begin{align}
\begin{split}
%n = & \frac{1}{2}N\\
	E\left( u_m[u_{m+n}=0] \right)= & E\left( u_m[u_{m+\frac{N}{2}}=0]\right) \\
	 = &E\left( -u_{m+\frac{N}{2}}[u_{m+\frac{N}{2}}=0]\right) = 0.
\end{split}
\end{align}
Thus, $E\left( u_m[u_{m+n}=0] \right) = 0$ for each of the possible values of $m$ and $n$.
Now consider $r_e[n]=E\left( [u_m=0][u_{m+n}=0]\right)$.
If $n \not\in \mathcal{M}$, $u_m$ and $u_{m+n}$ are statistically independent  and
\begin{align}
\begin{split}
	E\left( [u_m=0][u_{m+n}=0] \right) = & E([u_m=0])E([u_{m+n}=0])
	=  \frac{1}{9} \,.
\end{split}
\end{align}
If $n=0$, 
\begin{align}
\begin{split}
	E\left( [u_m=0][u_{m+n}=0] \right) & = E([u_m=0]^2) \\
	& =E([u_{m}=0]) = \frac{1}{3} \,.
\end{split}
\end{align}
Given that $u_m = -u_{m+\frac{N}{2}}$ if $n=\frac{N}{2}$, 
\begin{align}
\begin{split}
	E\left( [u_m=0][u_{m+\frac{N}{2}}=0] \right)= \frac{1}{3} \,.
\end{split}
\end{align}
Finally when $n = \frac{1}{6}N, \frac{1}{3}N,  \frac{2}{3}N, \frac{5}{6}N$,
$u_m$ and $u_{m+n}$ can not be zero at the same time given that they belong to the same triplet obtained by permuting elements in the vector $b_0$, so
\begin{align}
\begin{split}
	E\left( [u_m=0][u_{m+n}=0] \right)=0
\end{split}
\end{align}
results.
Thus $r_e[n] = \frac{1}{9}+ r_{E}[n]$ where
\begin{equation}
	r_E[n] = \left\{
\begin{array}{ll}
	\frac{2}{9} & n=0, \frac{N}{2} \\
	-\frac{1}{9} & n=\frac{1}{6}N, \frac{1}{3}N,  \frac{2}{3}N, \frac{5}{6}N \\
	0 & \mbox{otherwise}
\end{array}
	\right. \,.
\end{equation}
Thus, from (\ref{corre}),
\begin{align}
\begin{split}
	r_Y[n] = & \frac{2}{3}(a_0-\beta)\beta+ \beta^2+ \frac{1}{9}(a_0-\beta)^2 \\
	+ & \alpha^2 E(u_mu_{m+n}) + (a_0-\beta)^2r_E[n] \\ 
	= & \left(\frac{2}{3}\beta+ \frac{1}{3}a_0\right)^2+\alpha^2 E(u_mu_{m+n}) + 
	(a_0-\beta)^2r_E[n] \,.
\end{split}
\label{eq:r_Y}
\end{align}
Therefore, the autocorrelation sequence contains a constant term, the original autocorrelation sequence multiplied by a gain and an additional autocorrelation sequence representing the error term.

By taking the DFT of $r_E[n]$ we obtain
\begin{align}
\begin{split}
	R_E[k] = & \sum_{n=0}^{N-1}r_E[n]e^{-j2\pi k \frac{n}{N}} \\ 
	= & 
%	= \sum_{n=0}^{N-1}r_E[n]e^{-j2\pi k \frac{n}{N}}
	-\frac{1}{9} \left( \sum_{n=0}^{5}e^{-j2\pi k \frac{n}{6}}-1-e^{-j\pi k} \right) 
	+\frac{2}{9} (1+e^{-j\pi k}) \\
	= & -\frac{1}{9}\sum_{n=0}^{5}e^{-j2\pi k \frac{n}{6}} 
	+ \frac{1}{3}\{1+(-1)^k\}, \, k=0, \ldots, N-1 \,.
\end{split}
\end{align}
Since
\[
	\sum_{n=0}^{5}e^{-j2\pi k \frac{n}{6}} = \left\{
	\begin{array}{ll}
		6 & k = 6 m , \mbox{with $m$ integer}\\
		0 & \mbox{otherwise}
	\end{array}
	\right.
\]
we obtain
\begin{align}
\label{eq:R_E}
\begin{split}
	R_E[k] = &
%	= \sum_{n=0}^{N-1}r_E[n]e^{-j2\pi k \frac{n}{N}}
%	-\frac{1}{9} \left( \sum_{n=0}^{5}e^{-j2\pi k \frac{n}{6}}-1-e^{-j\pi k} \right) 
%	+\frac{2}{9} (1+e^{-j\pi k}) \\
	 -\frac{2}{3}[(k\mod 6) = 0]  + \frac{1}{3}\{ 1+(-1)^k\} \\
	 =  & \frac{2}{3}(\delta[k \bmod 2]-\delta[k \bmod 6]) \,.
\end{split}
\end{align}
The theorem follows by substituting \eqref{eq:R_E} into the DFT of \eqref{eq:r_Y}.

% \balance

%\section*{Acknowledgement}
%TODO
%This work was supported in part by the Fund for Scientific Research (FWO-Vlaanderen), by the Flemish Government (Methusalem), the Belgian Government through the Inter university Poles of Attraction (IAP VII) Program, and by the ERC advanced grant SNLSID, under contract 320378.
% the value of the ADC transition levels. 

\balance


\begin{thebibliography}{123}

\bibitem{MirriEtAl2004} D. Mirri, F. Filicori, G. Iuculano and G. Pasini, ``A nonlinear dynamic model for performance analysis of large-signal amplifiers in communication systems,'' \emph{IEEE Trans. Instrum. Meas.}, vol. 53, no. 2, pp. 341-350, April 2004.

\bibitem{Pintelon&Schoukens2012} R. Pintelon and J. Schoukens, \emph{System identification: a frequency domain approach}, John Wiley and Sons, Inc., Hoboken, NJ, 2$^{nd}$ edition, 2012.

\bibitem{Barker&Godfrey1999} H. A. Barker and K. R. Godfrey, ``System identification with multi-level periodic perturbation signals," \emph{Control Engineering Practice}, vol. 7, no. 6, pp. 717-726, June 1999.

\bibitem{BarkerEtAl2005} H. A. Barker, A. H. Tan and K. R. Godfrey, ``The design of ternary perturbation signals for linear system identification." \emph{16th IFAC World Congress}, Prague, 3-8 July 2005, Paper Tu-M13-TO/6.

\bibitem{DeAngelisEtAl_I2MTC2016} A. De Angelis, J. Schoukens, K. R. Godfrey and P. Carbone, ``Practical synthesis of ternary sequences for system identification,'' \emph{IEEE Int. Instrumentation and Measurement Technology Conf. (I2MTC)}, Taipei, Taiwan, May 23--26, 2016.

\bibitem{ZhouEtAl2016} T. Zhou, C. Tao, S. Salous, L. Liu and Z. Tan, ``Implementation of an LTE-Based Channel Measurement Method for High-Speed Railway Scenarios,'' \emph{IEEE Trans. Instrum. Meas.}, vol. 65, no. 1, pp. 25-36, Jan. 2016.

\bibitem{GeversEtAl2015} M. Gevers, P. Gebhardt, S. Westerdick, M. Vogt and T. Musch, ``Fast Electrical Impedance Tomography Based on Code-Division-Multiplexing Using Orthogonal Codes,'' \emph{IEEE Trans. Instrum. Meas.}, vol. 64, no. 5, pp. 1188-1195, May 2015.

\bibitem{BouchaalaEtAl2015} D. Bouchaala, E. Mekki, T. G{\"u}nther, P. B{\"u}schel, O. Kanoun and N. Derbel, "Study of excitation signals parameters for portable bioimedical devices," \emph{IEEE Int. Instrumentation and Measurement Technology Conf. (I2MTC) Proceedings}, Pisa, 2015, pp. 784-788.

\bibitem{GeerardynEtAl2013} E. Geerardyn, Y. Rolain and J. Schoukens, ``Design of Quasi-Logarithmic Multisine Excitations for Robust Broad Frequency Band Measurements,'' \emph{IEEE Trans. Instrum. Meas.}, vol. 62, no. 5, pp. 1364-1372, May 2013.

\bibitem{Friese1997} M. Friese, ``Multitone signals with low crest factor,'' in \emph{IEEE Transactions on Communications}, vol. 45, no. 10, pp. 1338-1344, Oct. 1997.

\bibitem{YangEtAl2015} Yang Y, Zhang F, Tao K, Sanchez B, Wen H and Teng Z. ``An improved crest factor minimization algorithm to synthesize multisines with arbitrary spectrum,'' \emph{Physiol. Meas.}, 2015 May;36(5):895-910.

\bibitem{OngEtAl2012} M. S. Ong, Y. C. Kuang, P. S. Liam and M. P. L. Ooi, ``Multisine With Optimal Phase-Plane Uniformity for ADC Testing,'' \emph{IEEE Trans. Instrum. Meas.}, vol. 61, no. 3, pp. 566-578, March 2012.

\bibitem{Vora&Satish2011} S. C. Vora and L. Satish, ``ADC Static Nonlinearity Estimation Using Linearity Property of Sinewave,'' \emph{IEEE Trans. Instrum. Meas.}, vol. 60, no. 4, pp. 1283-1290, April 2011.

\bibitem{VanMoer&Rolain2008} W. Van Moer and Y. Rolain, ``Multisine Calibration for Large-Signal Broadband Measurements,'' \emph{IEEE Trans. Instrum. Meas.}, vol. 57, no. 7, pp. 1478-1483, July 2008.

\bibitem{Tan&Godfrey2002} A. H. Tan and K. R. Godfrey, ``The generation of binary and near-binary pseudorandom signals: an overview,'' \emph{IEEE Trans. Instrum. Meas.}, vol. 51, no. 4, pp. 583-588, Aug 2002.

\bibitem{WongEtAl2013} H. K. Wong, J. Schoukens and K. R. Godfrey, ``Design of Multilevel Signals for Identifying the Best Linear Approximation of Nonlinear Systems,'' \emph{IEEE Trans. Instrum. Meas.}, vol. 62, no. 2, pp. 519-524, Feb. 2013.

\bibitem{TanEtAl2005} A. H. Tan, K. R. Godfrey and H. A. Barker, ``Design of computer-optimized pseudorandom maximum length signals for linear identification in the presence of nonlinear distortions,'' \emph{IEEE Trans. Instrum. Meas.,} vol. 54, no. 6, pp. 2513-2519, Dec. 2005.

\bibitem{GodfreyEtAl1999} K. R. Godfrey, H. A. Barker and A. J. Tucker, ``Comparison of perturbation signals for linear system identification in the frequency domain," in \emph{IEE Proceedings - Control Theory and Applications}, vol.146, no.6, pp.535-548, Nov 1999.

\bibitem{Tan} A. H. Tan, ``Direct synthesis of pseudo-random ternary perturbation signals with harmonic multiples of two and three suppressed," \emph{Automatica}, vol. 49, no. 10, pp. 2975-2981, October 2013.

\bibitem{Galton&Carbone1995} I. Galton and P. Carbone, ``A rigorous error analysis of D/A conversion with dynamic element matching," in \emph{IEEE Transactions on Circuits and Systems II}, vol.42, no.12, pp. 763-772, Dec 1995.

\bibitem{Agilent33220ADatasheet} \emph{Keysight 33220A arbitrary waveform generator user manual}, Online. Available: http://www.keysight.com. 

\bibitem{SchoukensEtAl2003} J. Schoukens, Y. Rolain, G. Simon and R. Pintelon, ``Fully automated spectral analysis of periodic signals," in \emph{IEEE Instrumentation and Measurement Technology Conf. (IMTC)}, pp.299-302, 2002.

\bibitem{IEEE1241} IEEE Instrumentation $\&$ Measurement Society, \emph{IEEE Standard for Terminology and Test Methods for Analog-to-Digital Converters}, IEEE Std 1241-2010, Jan. 14, 2011.

\bibitem{BarkerEtAl2004} H. A. Barker, A. H. Tan and K. R. Godfrey, ``Design of multilevel perturbation signals with harmonic properties suitable for nonlinear system identification," in \emph{IEE Proceedings - Control Theory and Applications}, vol. 151, no. 2, pp. 145-151, 23 March 2004.


\end{thebibliography}
\end{document}